\documentclass[3p,preview]{elsarticle} 

\usepackage{graphicx}
\usepackage{dcolumn}
\usepackage{bm}
\usepackage{color}
\usepackage{amsfonts}
\usepackage{amssymb}
\usepackage{amsmath}
\usepackage{amsthm} 
\usepackage{enumerate}

\biboptions{sort&compress}
\theoremstyle{plain}
\newtheorem{thm}{Theorem}

\newtheorem{prp}{Proposition}
\newtheorem{cor}{Corollary}

\theoremstyle{definition}

\journal{Applied Mathematical Letters}

\begin{document}

\begin{frontmatter}

\title{Simultaneous diagonalization via congruence: Nondegenerate families of symmetric matrices }

\author[unifei]{L. G. Gomes\corref{cor1}}
\ead{lggomes@unifei.edu.br}
\cortext[cor1]{Corresponding author}

\address[unifei]{Institute of Mathematics and Computation, Federal University of Itajubá - Brazil}

\begin{abstract}
In this letter, we give necessary and sufficient conditions for the real symmetric matrices $A_1 , \ldots , A_m$ to ensure they are simultaneously diagonalizable via congruence.   
\end{abstract}


\end{frontmatter}

\section{Introduction}

Simultaneous diagonalization via congruence of a general family of matrices is an important issue in applied mathematics that arises under specific applications in many areas, as diverse as signal processing \cite{DeLathauwer2006}, cosmology\footnote{This work arouse when investigating simultaneous diagonalizability of Bianchi models in cosmology\cite{Gomes}, which triggered the process that culminated in this manuscript.} \cite{Gomes, LQC_Bianchi2023}, acoustic and structural engineering \cite{Garvey2002}, and optimization \cite{BenTal2014}, so much so that in Ref. \cite{Hiriart} it has been cast as an important open question in applied mathematics. It is a well studied topic with sparse and disconnected results along the $19^{th}$ and $20^{th}$ centuries (see \cite{Uhlig} for a review of the achievements on the subject at that time), which in general involved just two symmetric forms \cite{Becker,Greub, Wonen, Uhlig2}. Only in the last decade a significant contribution have appeared involving a general number of forms \cite{Jiang}. In this manuscript, we generalise this last result.

To state the problem, take a real vector space $E$ of dimension $n$ with a family $\mathcal{F}=\{\gamma_t \, | \, t \in \mathcal{I} \}$ of symmetric bilinear forms on it, where $\mathcal{I}$ is any set of indexes. Any basis $\mathfrak{B}=\{e_1,\ldots,  e_n\}$ of $E$ gives rise to a family of symmetric $n\times n$ matrices with entries $\gamma_{ij}(t) := \gamma_t(e_i,e_j)$:    
\begin{equation}
[\gamma_t]_\mathfrak{B} := 
\left[ 
\begin{array}{ccc}
\gamma_{11}(t) & \gamma_{12}(t) & \ldots \\
\gamma_{12}(t) & \gamma_{22}(t) & \ldots \\
\vdots & \vdots & \ddots
\end{array}
\right] \quad . 
\end{equation}
If  $[\gamma_t]_\mathfrak{B}$ is in diagonal form for every $t \in \mathcal{I}$, i.e.,  $\gamma_t(e_i,e_j) =0$ for $i \ne j$, we say $\mathfrak{B}$ diagonalizes $\{\gamma_t \}$. This family is called simultaneously diagonalizable if it is diagonalized by a basis of $E$. In a simpler way, this is equivalent to the existence of a $n\times n$ invertible matrix $Q$ diagonalizing the entire family of symmetric matrices $[\gamma_t]$, that is,
\begin{equation}
 Q^\dagger \, [\gamma_t] \, Q  = 
\left[ 
\begin{array}{cccc}
\lambda_{1}(t) & 0 & \ldots & 0 \\
0 & \lambda_{2}(t) & \ldots & 0 \\
\vdots & \vdots & \ddots & \vdots \\
0 & 0 & \ldots & \lambda_{n}(t) 
\end{array}
\right] \qquad \forall \, t \, \in \mathcal{I} ,
\end{equation}
where '$\dagger$' stands for the transpose of a matrix. In the language of matrices, this is referred to as simultaneous diagonalization via congruence. The main issue associated with this concept is to find simple algebraic relations which are necessary and sufficient for ensuring it. For this, we put forth Theorem \ref{Thm:SimultaneousDiagonalization}, that in matrix terms can be stated as 
\begin{thm}\label{Th:SimultaneousDiagonalizationMatrix} 
Suppose $A_1  , \ldots , A_m$ are real $n \times n$ symmetric matrices and $A_0$ is any invertible matrix obtained as a linear combination of them. They are simultaneously diagonalizable via congruence if, and only if, the following two conditions are satisfied:
\begin{itemize}
\item[(SD1)] Each $G_i := A_0^{-1} \, A_i$ is a diagonalizable matrix, that is, its Jordan canonical form is diagonal.
%
\item[(SD2)] For each matrix $b_{ij}:= A_i \, A_0^{-1} \, A_j$,  we have $b_{ij}=b_{ji}$.
\end{itemize}
\end{thm}
Indeed, condition $(SD1)$ tells us that the pair $\{A_0, A_i\}$ is simultaneous diagonalizable for each $i$ \cite{Greub}, while  condition $(SD2)$ ensures that a set of matrices simultaneously diagonalizable to a common one can be simultaneously diagonalized as a whole. 
If we specialize to a positive definite pencil, we may restate the previous theorem in an even simpler form, which can be seen as the coordinate free equivalent to the Theorem 3.3 in \cite{Jiang}: 
\begin{cor}\label{Cor:SimultaneousDiagonalizationMatrixDefinite} 
If $A_0$ is positive (negative) definite, then the matrices $A_1  , \ldots , A_m$  are simultaneously diagonalizable via congruence if, and only if, $b_{ij}=b_{ji}$ for each  $i,j=1, \ldots , m$.
%
%
\end{cor}

In the following, we prove the facts stated above in a coordinate free manner, so that a symmetric matrix is just a representation of a symmetric bilinear form in a chosen basis and congruence arises as a change of coordinates. 

\section{Simultaneous diagonalization for non-degenerate families of symmetric forms}\label{Sec:MainThm}
For a symmetric bilinear form $\gamma$, when the induced linear map $\gamma^\flat$ is an isomorphism between $E$ and its dual space $E^*$, where $\gamma^\flat(u):= \gamma(u,\cdot)$, we say that $\gamma$ is non-degenerate. In this case, its inverse is denoted by $\gamma^{\#}: E^* \to E$.  We assume that the family $\mathcal{F}=\{\gamma_t \, | \, t \in \mathcal{I} \}$ of symmetric bilinear forms is non-degenerate, by which we mean that it has a non-degenerate element $\gamma_{t_0}$ for some $t_0 \in \mathcal{I}$. Therefore,  for each $t \in \mathcal{I}$ we can build a linear operator $G_{t_0}(t)$ on $E$ by  
\begin{equation}
 G_{t_0}(t):= \gamma^{\#}_{t_0} \circ \gamma^\flat_{t} \quad .
\end{equation}
Its  matrix representation in a basis $\mathfrak{B}$ of $E$ is simply
\begin{equation}
[G_{t_0}(t)]_\mathfrak{B} = [\gamma_{t_0}]^{-1}_\mathfrak{B}\, [\gamma_t]_\mathfrak{B} \, .
\end{equation}
It is straightforward to verify that (\cite{Greub}, sec. 9.3)
\begin{prp}\label{Prop:DiagonalizacaoG}
The two symmetric forms $\gamma_{t_0}$ and $\gamma_{t}$ are simultaneously diagonalizable if, and only if, $G_{t_0}(t)$ is a diagonalizable operator, that is, there exists a basis for $E$ of eigenvectors of $G_{t_0}(t)$.
\end{prp}
For $t,s \in \mathcal{I}$ we can define a bilinear form $b_{t_0}(s,t)$ by the relation
\begin{equation}
 b_{t_0}(s,t)\cdot (u,v):= \gamma_s(u, G_{t_0}(t)\cdot v) \, .
\end{equation}
In its matrix representation,
\begin{equation}
[b_{t_0}(s,t)]_\mathfrak{B} = [\gamma_s]_\mathfrak{B} \, [\gamma_{t_0}]^{-1}_\mathfrak{B}\, [\gamma_t]_\mathfrak{B} \, .
\end{equation}
In general, $b(s,t)$ is not symmetric. In fact, we have:
\begin{prp}\label{Prop:SimetriaB}
The bilinear form $b_{t_0}(s,t)$ is symmetric if, and only if,  $b_{t_0}(s,t)= b_{t_0}(t,s)$.
\end{prp}
\proof
The equivalence above follows from  
\begin{eqnarray*}
[b_{t_0}(s,t)]_\mathfrak{B}^\dagger & = & \left([\gamma_s]_\mathfrak{B} \, [\gamma_{t_0}]^{-1}_\mathfrak{B} \, [\gamma_t]_\mathfrak{B}\right)^\dagger \\
& = & [\gamma_t]_\mathfrak{B}^\dagger \, [\gamma_{t_0}]^{-\dagger}_\mathfrak{B} \, [\gamma_s]^\dagger_\mathfrak{B} \\
& = & [\gamma_t]_\mathfrak{B} \, [\gamma_{t_0}]^{-1}_\mathfrak{B} \, [\gamma_s]_\mathfrak{B} \\
& = & [b_{t_0}(t,s)]_\mathfrak{B} \, .
\end{eqnarray*}
\qed

The main result in this paper is the following theorem:
\begin{thm} \label{Thm:SimultaneousDiagonalization} 
In  a finite real vector space, given a family $\mathcal{F}=\{\gamma_t \, | \, t \in \mathcal{I} \}$ of symmetric bilinear forms with a non-degenerate element $\gamma_{t_0}$, where $\mathcal{I}$ is any set of indexes, it is simultaneously diagonalizable if, and only if, the following two conditions are satisfied:
\begin{itemize}
\item[(SD1)] For each $t \in \mathcal{I}$,  $G_{t_0}(t)$ is a diagonalizable operator on $E$;
\item[(SD2)] For each $s,t \in \mathcal{I}$, the bilinear form $b_{t_0}(s,t)$ is symmetric.
\end{itemize}
\end{thm}
\proof
If $\{\gamma_t \, | \, t \in \mathcal{I} \}$ is simultaneously diagonalizable, then $b_{t_0}(s,t)$ is symmetric for each pair $(s,t)$, as can be seen in a simultaneously diagonal matrix form of the elements of such family. Therefore, the two properties above follows trivially.  

On the other hand, let the properties $(SD1)$ and $(SD2)$ hold for $\mathcal{F}$. By Proposition \ref{Prop:DiagonalizacaoG}, there exists a basis for $E$ of eigenvectors of $G_{t_0}(t)$. In particular, each of its eigenvalues is real. Such operator decomposes the vector space in their eigenspaces
\begin{equation}\label{Eq:Decomposicao}
 E= E_1 \oplus \ldots \oplus E_r \, ,
\end{equation}
where each $E_i$ is associated with a different eigenvalue $\lambda_i$ of $G_{t_0}(t)$. Hence 
\begin{equation}
G_{t_0}(t) \cdot u_i = \lambda_i \, u_i \, ,
\end{equation}
whenever $u_i \in E_i$. The first important property of this decomposition is that it is also $\gamma_s$-ortogonal for any $s \in \mathcal{I}$. This follows from the symmetry of $b_{t_0}(s,t)$, for taking $u_\ell \in E_\ell$ and noting that 
\begin{eqnarray*}
\lambda_j \, \gamma_s(u_i, u_j) & = & \gamma_s( u_i, G_{t_0}(t) \cdot u_j) \\
    & = & b_{t_0}(s,t) \cdot (u_i, u_j) \\
    & = & b_{t_0}(s,t) \cdot (u_j, u_i) \\
    & = & \gamma_s( u_j, G_{t_0}(t) \cdot u_i) \\
    & = & \lambda_i \, \gamma_s(u_i, u_j) \, ,
\end{eqnarray*}
we have 
\begin{eqnarray} \label{Eq:DecomposicaoGammaOrtogonal}
i \ne j & \Longrightarrow & \gamma_s( u_i, u_j) =0 \, . 
\end{eqnarray}
Furthermore, we can also prove that
\begin{equation} \label{Eq:DecomposicaoPertinenciaPartes}
u \in E_i \quad \iff \quad \quad \forall\, j \ne i \quad \forall \, u_j \in E_j \, :\,  \gamma_{t_0}( u, u_j) =0  \, . 
\end{equation}
The implication follows from the formula (\ref{Eq:DecomposicaoGammaOrtogonal}). Conversely, if $u \notin E_i$ then $u = u_1 + \ldots + u_r$ with $u_\ell \in E_\ell$, $\ell=1, \ldots , r$, and $u_j \ne 0$ for some $j \ne i$. Hence, since $\gamma_{t_0}$ is non-degenerate and by (\ref{Eq:DecomposicaoGammaOrtogonal}) the decomposition of  $E$ is $\gamma_{t_0}$-orthogonal, there must be an element $v_j \in E_j$ such that $\gamma_{t_0}(u_j, v_j)=1$.%
\footnote{If this was not true, we would conclude that $\gamma_{t_0}(u_j,\cdot)=0$, contradicting the nondegeneracy of $\gamma_{t_0}$.} 
Therefore, $\gamma_{t_0}(u, v_j)=\gamma_{t_0}(u_j, v_j)=1$, proving the statement in (\ref{Eq:DecomposicaoPertinenciaPartes}). 

At this point we have the right conditions to assert that each $E_i$ is invariant under $G_{t_0}(s)$, for any $s \in \mathcal{I}$. In order to prove it, note that $b_{t_0}(t_0,s)=\gamma_s$ and as we take $u_\ell \in E_\ell$, $\ell=1, \ldots, r$, $i \ne j$, we get
\begin{eqnarray*}
\gamma_{t_0}(u_j, G_{t_0}(s) \cdot u_i) & = & b_{t_0}(t_0,s) \cdot (u_j, u_i) \\
    & = & \gamma_s( u_j, u_i) \\
    & = & 0 \, .
\end{eqnarray*}
The assertion (\ref{Eq:DecomposicaoPertinenciaPartes}) implies that $G_{t_0}(s) \cdot u_i \in E_i$ whenever $u_i \in E_i$, checking our statement.

Now we can construct our simultaneously orthogonal basis. For this, take $t' \in \mathcal{I}$ different from $t$ and $t_0$.%
\footnote{If there is none, this theorem would reduce to Proposition \ref{Prop:DiagonalizacaoG}.}
By hypothesis, each eigenvalue of $G_{t_0}(t')$ is real and there is a basis of $E$ composed by the associated eigenvectors. Since each subspace $E_i$ in the decomposition (\ref{Eq:Decomposicao}) is invariant under $G_{t_0}(t')$, if we have $G_{t_0}(t') \cdot u' = \lambda' \, u'$, then $G_{t_0}(t') \cdot u_\ell' = \lambda' \, u_\ell'$ for each $\ell=1, \ldots ,r$, as $u'=u_1' + \ldots + u_r'$. Therefore we conclude that the restriction of $G_{t_0}(t')$ to $E_i$ is also a diagonalizable operator on $E_i$. Now the theorem easily follows by  induction in the dimension of the space $E$. Indeed, probing it a little bit deeper, if we proceed like above over and over again with different values in $\mathcal{I}$, we will stop in a final decomposition  
\begin{equation}\label{Eq:DecomposicaoFinal}
 E= \widetilde{E}_1 \oplus \ldots \oplus \widetilde{E}_{R} \, ,
\end{equation}
such that for any $s \in \mathcal{I}$, it is $\gamma_s$-orthogonal and 
\begin{equation}
 G_{t_0}(s) \cdot \widetilde{u}_i = \lambda_i(s) \, \widetilde{u}_i  
\end{equation}
for each $\widetilde{u}_i \in \widetilde{E}_i$, $i= 1, \ldots , R$.
The decomposition (\ref{Eq:DecomposicaoFinal}) means that each $\gamma_s$ is given by the orthogonal sum
\begin{equation}\label{Eq:DecomposicaoFinalFormas}
 \gamma_s = \gamma_s^{(1)} \oplus \ldots \oplus \gamma_s^{(R)} \, ,
\end{equation}
with each $\gamma_s^{(i)}$ a symmetric bilinear form in $\widetilde{E}_{i}$ satisfying
\begin{equation}\label{Eq:DecomposicaoFinalFormas2}
 \gamma_s^{(i)} = \lambda_i(s) \, \gamma_{t_0}^{(i)} \, .
\end{equation}
Therefore, any orthogonal basis for $\gamma_{t_0}$ in which each of its elements is in some subspace $\widetilde{E}_{i}$ is a simultaneous orthogonal basis for the family $\mathcal{F}$. 
\qed

We recall that $\gamma$ is positive definite if $\gamma (v,v) >0$ for $v \neq 0$, and negative definite if $-\gamma$ is positive definite.  It is clear that every definite symmetric form is non-degenerate. Furthermore, it is a common knowledge in linear algebra that, if $\gamma_{t_0}$ is positive (negative) definite, then $G_{t_0}(t)$ is a diagonalizable operator on $E$, for it is symmetric with respect to the inner product $\gamma_{t_0}$, that is, $\gamma_{t_0}(G_{t_0}(t)\cdot u, v) = \gamma_{t_0}( u,G_{t_0}(t)\cdot v)$. This follows from the relation
\[
 \gamma_{t_0}(G_{t_0}(t)\cdot u, v) = \gamma_{t_0}^\flat(G_{t_0}(t)\cdot u)\cdot v 
 = (\gamma_{t_0}^\flat \circ \gamma_{t_0}^{\#} \circ \gamma_{t}^\flat) (u) \cdot v 
 = \gamma_{t}(u,v) \, .
\]
It implies that $\gamma_{t_0}$ and $\gamma_{t}$ are simultaneously diagonalizable. Therefore, the hypothesis in the main Theorem \ref{Thm:SimultaneousDiagonalization} can be relaxed and we get 
\begin{cor} \label{Cor:SimultaneousDiagonalization}
If $\gamma_{t_0}$ is positive or negative definite for some $t_0 \in \mathcal{I}$, then the family $\mathcal{F}=\{\gamma_t \, | \, t \in \mathcal{I} \}$ is simultaneously diagonalizable if, and only if,  for each $s,t \in \mathcal{I}$, the bilinear form $b_{t_0}(s,t)$ is symmetric.    
\end{cor}

\vspace{6mm}
 
{\it Acknowledgements:} The author is thankful for the support from FAPEMIG, project number RED-00133-21.

\bibliography{Ref_Diagonalization}

\end{document}